\documentclass[letterpaper,11pt]{article}
\usepackage{url}
\usepackage{amssymb,amsmath,amsfonts,amsthm}
\usepackage{fullpage}
\usepackage[usenames]{color}

\newcommand{\polylog}{{\mathrm{polylog}}}

\DeclareMathSymbol{\qedsymb} {\mathord}{AMSa}{"04}

\newcommand{\eps}{\varepsilon}
\renewcommand{\epsilon}{\varepsilon}

\newcommand{\ceil}[1]{\left\lceil #1 \right\rceil}
\newcommand{\floor}[1]{\left\lfloor #1 \right\rfloor}

\newcommand{\oct}{\quad\quad}                                   

\newcommand{\norm}[1]{\left\lVert #1 \right\rVert}

\newcommand{\TLE}{\mathrm{TLE}}

\newcommand{\GF}{\mathrm{GF}}
\newcommand{\dlog}{\mathrm{dlog}}
\newcommand{\E}{\mathbf{E}}

\newcommand{\Var}{\mathbf{Var}}
\renewcommand{\Pr}{\mathbf{Pr}}
\renewcommand{\mod}{\hbox{ mod }}

\newcommand{\EquationName}[1]{\label{eq:#1}}

\newcommand{\LemmaName}[1]{\label{lem:#1}}

\newcommand{\RemarkName}[1]{\label{rem:#1}}
\newcommand{\SectionName}[1]{\label{sec:#1}}
\newcommand{\TheoremName}[1]{\label{thm:#1}}
\newcommand{\FigureName}[1]{\label{fig:#1}}

\newcommand{\Equation}[1]{Eq.\:\eqref{eq:#1}}

\newcommand{\Lemma}[1]{Lemma~\ref{lem:#1}}

\newcommand{\Remark}[1]{Remark~\ref{rem:#1}}
\newcommand{\Section}[1]{Section~\ref{sec:#1}}
\newcommand{\Theorem}[1]{Theorem~\ref{thm:#1}}
\newcommand{\Figure}[1]{Figure~\ref{fig:#1}}

\newtheorem{theorem}{Theorem}

\newtheorem{lemma}[theorem]{Lemma}
\newtheorem{remark}[theorem]{Remark}

\newcommand{\proofbelow}{3pt}
\newcommand{\afterproof}{\hfill $\blacksquare$ \par \vspace{\proofbelow}}
\newcommand{\aftersubproof}{\hfill $\Box$ \par \vspace{\proofbelow}}
\renewenvironment{proof}{\noindent\textbf{Proof.}\,}{\afterproof}

\newcommand{\poly}{\mathop{{\rm poly}}}
\renewcommand{\th}{\ifmmode{^{\textrm{th}}}\else{\textsuperscript{th}\ }\fi}

\newcommand{\eqdef}{\mathbin{\stackrel{\rm def}{=}}}
\newcommand{\comment}[1]{}

\begin{document}

\author{Jelani Nelson\footnotemark[2]\oct
  David P. Woodruff\footnotemark[3]}

\date{}

\footnotetext[1]{MIT Computer Science and Artificial Intelligence
  Laboratory. \texttt{minilek@mit.edu}. Supported by a National
  Defense Science and Engineering Graduate (NDSEG) Fellowship. Much of this
  work was done while the author was at the IBM Almaden
  Research Center.}
\footnotetext[2]{IBM Almaden Research Center, 650 Harry Road, San
  Jose, CA, USA. \texttt{dpwoodru@us.ibm.com}.}

\title{A Near-Optimal Algorithm for L1-Difference}

\maketitle

\begin{abstract}
\thispagestyle{empty}
\noindent

We give the first $L_1$-sketching algorithm for integer vectors which
produces nearly optimal sized sketches in nearly linear time. This
answers the first open problem in the list of open problems 
from the 2006 IITK Workshop on Algorithms for Data Streams. 
Specifically, suppose Alice
receives a vector $x\in\{-M,\ldots,M\}^n$ and Bob receives
$y\in\{-M,\ldots,M\}^n$,
and the two parties share randomness.  Each party must output a short 
sketch of their vector such that a
third party can later quickly recover a $(1\pm\eps)$-approximation to
$||x-y||_1$ with $2/3$ probability given only the sketches.
We give a sketching algorithm which produces
$O(\eps^{-2}\log(1/\eps)\log(nM))$-bit sketches in $O(n\log^2(nM))$
time, independent of $\eps$. The
previous best known sketching algorithm for $L_1$ is due to
[Feigenbaum {\it et al.}, SICOMP 2002], which achieved the optimal
sketch length of
$O(\eps^{-2}\log(nM))$ bits but had a running time of
$O(n\log(nM)/\eps^2)$. Notice that our
running time is near-linear for every $\eps$, whereas for sufficiently
small values of $\eps$, the running time of the previous algorithm can
be as large as
quadratic. Like their algorithm, our sketching
procedure also yields a small-space, one-pass streaming algorithm
which works even if the
entries of $x,y$ are given in arbitrary order. 
\end{abstract}

\section{Introduction}\SectionName{intro}
Space and time-efficient processing of massive databases is a
challenging and important task in applications such as observational
sciences, product marketing, and monitoring large systems. Usually the
data set is distributed across several network devices, each receiving
a portion of the data as a stream. The devices must locally
process their data, producing a small sketch, which can then be
efficiently transmitted to other devices for further
processing. Although much work has focused on producing sketches of
minimal size for various problems, in practice the time efficiency to
produce the sketches is just as important, if not more so, than the
sketch size. 

In the 2006 IITK Workshop on Algorithms for Data Streams, the first
open question posed \cite{IITK} was to find a space- {\it and}
time-efficient algorithm for $L_1$-difference computation. Formally,
there are two parties, Alice and Bob, who have vectors $x, y \in \{-M,
\ldots, M\}^n$ and wish to compute sketches $s(x)$
and $s(y)$, respectively, so that a third party can quickly recover a
value $Z$ with $(1-\eps)\norm{x-y}_1 \leq Z \leq
(1+\eps)\norm{x-y}_1$. Here, $\norm{x-y}_1 = \sum_{i=1}^n |x-y|_i$
denotes the $L_1$-norm of the vector $x-y$. The third party should
succeed with probability at least $2/3$ over the randomness of Alice
and Bob (this probability can be amplified by
repeating the process and taking the median).

The original motivation \cite{FKSV02} for the $L_1$-difference problem is
Internet-traffic monitoring. As packets travel through
Cisco routers, the NetFlow software \cite{CN98} produces summary
statistics of groups of packets with the same source and destination
IP address. Such a group of packets is known as a {\it flow}. At the
end of a specified time period, a router assembles sets of values $(s,
f_t(s))$, where $s$ is a source-destination pair, and $f_t(s)$ is the
total number of bytes sent from the source to the destination in time
period $t$. The $L_1$-difference between such sets assembled during
different time periods or at different routers indicates differences
in traffic patterns.

The ability to produce a short sketch summarizing the set of values
allows a central control and storage facility to later efficiently
approximate the $L_1$-difference between the sketches that it
receives. The routers producing the sketches cannot
predict which source-destination pairs they will receive, or in which
order. Since the routers can transmit their sketches and updates
to their sketches to the central processing facility in an arbitrarily
interleaved manner, it is essential that the $L_1$-difference
algorithm support arbitrary permutations of the assembled sets of
values. Because of the huge size of the packet streams, it is also
crucial that the computation time required to produce the sketches
be as small as possible.

The first algorithm for this problem is due to Feigenbaum {\it et al.}
\cite{FKSV02}, and achieves a sketch of size $O(\eps^{-2} \log (nM))$
bits with each party requiring $O(n\log (nM)/\eps^2)$ processing
time (in word operations). Later, Indyk \cite{Indyk06} generalized this
to the
problem of estimating the $L_1$-norm of a general data stream with
an arbitrary number of updates to each coordinate. For the
$L_1$-difference problem, the space of Indyk's method is worse than
that of \cite{FKSV02} by a factor of $\log (n)$, while the time
complexity is similar. Recently, in \cite{KNW08} it was shown how to
reduce the space complexity of Indyk's method, thereby matching the
sketch size of \cite{FKSV02}. 
However, the time complexity per stream update 
is $\Omega(\eps^{-2})$. Also in \cite{KNW08}, a space lower
bound of $\Omega(\eps^{-2} \log (nM))$ was shown for the
$L_1$-difference problem for nearly the full range of interesting
values for $\eps$, thereby showing that the space complexity of
the algorithms of \cite{FKSV02, KNW08} are optimal.

While the space complexity for $L_1$-difference is settled, there are
no non-trivial lower bounds for the time complexity. Notice that the
$\eps^{-2}$ factor in the processing time can be a severe drawback in
practice, and can make the difference between setting the
approximation quality to $\eps = .1$ or to $\eps = .01$. Indeed, in
several previous works (see the references in Muthukrishnan's book
\cite{Muthu}, or in Indyk's course notes \cite{IndykCourse}), the main
goal was to reduce the dependence on $\eps$ in the space
and/or time complexity. This raises
the question, posed in the IITK workshop, as to whether this
dependence on $\eps$ can be improved for $L_1$-difference. As a first step, Cormode and
Ganguly \cite{CG07} show that if one increases the sketch size to
$\eps^{-3} \polylog (nM)$, then it is possible to achieve processing
time $n \cdot \polylog (nM)$, thereby removing the dependence on
$\eps$. Their algorithm even works for $L_1$-norm estimation of
general data streams, and not just for $L_1$-difference.
However, for reasonable values of $nM$, the sketch size is
dominated by the $\eps^{-3}$ term, which may be prohibitive in
practice, and is sub-optimal.

In this paper we show how to achieve a near-optimal $O(\eps^{-2}
\log(1/\eps)\log (nM))$ sketch size, while simultaneously achieving
a near linear $O(n \log^2 (nM))$ processing time, independent of $\eps$. Notice our
space is only a
factor of $\log(1/\eps)$ more than the lower bound. The time for a third
party to recover a $(1 \pm \eps)$-approximation to the $L_1$-difference,
given the sketches,
is nearly linear in the sketch size.
Furthermore, our
sketching procedure naturally can be implemented as a one-pass
streaming algorithm over an adversarial ordering of the coordinates of
$x,y$ (this was also true of previous algorithms). 
Thus, up to small
factors, we
resolve the first open question of \cite{IITK}.
We henceforth
describe our sketching procedure as a streaming algorithm.

While in \cite{CG07} and \cite{IITK} it is suggested to use the
techniques of \cite{bgks06} and \cite{IW} for estimating $L_p$ norms,
$p > 2$, which are themselves based on estimating coordinates of heavy
weight
individually and removing them, we do not follow this
approach. Moreover, the approach of \cite{FKSV02} is to embed
$L_1$-difference into $L_2^2$, then use the AMS sketch \cite{AMS99}
and range-summable hash functions they design
to reduce the processing time. We do not follow this approach either.

Instead, our first idea is to embed the $L_1$-difference problem into
$L_0$, the number of non-zero coordinates of the underlying vector 
(in this case $x-y$) presented as data stream. Such an
embedding has been used before, for example, in lower bounding the
space complexity of estimating $L_0$ in a data stream
\cite{IW03}.  Suppose for simplicity $x_i,y_i\ge 0$ for all $i\in
[n]$. Here the idea is for Alice to treat her input $x_i$ as a
set of distinct items $M(i-1)+1, \ldots, M(i-1) + x_i$, while Bob
treats his input $y_i$ as a set of distinct items $M(i-1) + 1,
\ldots, M(i-1) + y_i$. Then the size of the set-difference of these
two sets is $|x_i-y_i|$. Thus, if Alice inserts all of the set
elements corresponding to her coordinates as insertions into an
$L_0$-algorithm, while Bob inserts all of his elements as deletions,
the $L_0$-value
in the resulting stream equals $\norm{x-y}_1$. A recent
space-efficient algorithm for estimating $L_0$ with deletions
is given in \cite{KNW08}.

The problem with directly reducing to $L_0$ is that, while the
resulting space complexity is small, the processing time can be as
large as $O(nM)$ since we
must insert each set element into the $L_0$-algorithm. We overcome
this by developing a range-efficient $L_0$ algorithm, i.e. an
algorithm which allows updates to ranges at a time,
which works for streams coming out of our reduction by exploiting the
structure of ranges we update (all updated ranges
are of length at most $M$ and start at an index of the form $M(i-1)+1$).
We note that range-efficient $L_0$ algorithms have been
developed before \cite{BKS02,PavanTir07}, but those algorithms do
not allow deletions and thus do not suffice for our purposes.

At a high level, our algorithm works by sub-sampling by powers of $2$
the universe $[nM]$ arising out of our reduction to $L_0$. At each
level we keep a data structure of size
$O(\eps^{-2}\log(1/\eps))$ to summarize the items that are
sub-sampled at that level. We also maintain a
data structure on the side to handle the case when $L_0$ is
small, and we in parallel obtain a constant-factor approximation $R$
of the $L_1$-difference using \cite{FKSV02}. At the stream's end, we
give our estimate of the $L_1$-difference based on the summary data
structure living at the level where the expected number of universe
elements sub-sampled is $\Theta(1/\eps^2)$ (we can
determine this level knowing $R$).
As is the case in many previous streaming algorithms, the sub-sampling
of the stream can be
implemented using pairwise-independent hash functions. This allows us
to use a subroutine developed by Pavan and Tirthapura
\cite{PavanTir07} for quickly counting the number of universe elements
that
are sub-sampled at each of the $\log(nM)$ levels. Given these counts,
our summary data structures are such that we can update each one
efficiently.

Our summary data structure at a given level maintains $(x'-y')H$,
where $H$ is the parity-check matrix of
a linear error-correcting code, and $x',y'$ are the vectors derived
from $x,y$ by sub-sampling at that level. When promised that
$x',y'$ differ on few coordinates, we can treat $x'-y'$ as
a corruption of the encoding of the $0$ codeword then attempt to
decode to recover the ``error'' $x'-y'$. The decoding succeeds as long
as the minimum distance of the code is sufficiently high. This idea of
using error-correcting codes to
sketch vectors whose distance is promised to be small is known in the
error-correcting codes literature as {\em
  syndrome decoding} \cite{Wyner74}. Aside from error-correction,
syndrome decoding has also
found uses in cryptography: in the work of
\cite{BBR88,Smith07} to give a two-party protocol
for agreeing on a shared secret key when communicating over a noisy
channel, and in the work of \cite{FIMNSW06} as part of a private
two-party communication protocol for computing the Hamming distance
between two bitstrings.

For efficiency reasons, our implementation of the summary data
structure is mostly inspired by work of Feigenbaum {\it et al.}
\cite{FIMNSW06}. Given that
$x',y'$ differ on at most $k$ coordinates, they use the parity-check
matrix of a Reed-Solomon code of minimum distance $O(k)$.  Decoding
can then be done in time $O(k^2 + k\cdot \polylog(k)\log(n))$ using
an algorithm of Dodis {\it et al.} \cite{DORS08}. In our
application $k =
\Theta(\eps^{-2})$, and thus this recovery procedure is too slow for our
purposes.  To remedy this, we first hash the indices of $x',y'$
into $O(\eps^{-2}/\log(1/\eps))$ buckets with an
$O(\log(1/\eps))$-wise independent hash function, then in each bucket
we keep the product of the difference vector, restricted to the
indices mapped to that bucket, with the parity check matrix.
With constant probability, no bucket receives more than
$O(\log(1/\eps))$ indices where $x',y'$ differ.
We can thus use a Reed-Solomon code
with minimum distance only $O(\log(1/\eps))$, making the algorithm of
\cite{DORS08} fast enough for our purposes.

We note that our summary data structure in each level is in fact a
$k$-set structure, as defined by Ganguly \cite{Ganguly07}, that can be
used
to return a set of $k$ items undergoing insertions and deletions in a
data stream. While Ganguly's $k$-set structure uses near-optimal space
and has fast update time, it only works in the strict turnstile model
(i.e., it requires that each coordinate of $z = x-y$ is non-negative,
in which case the $L_1$-difference problem has a trivial solution:
maintain an $O(\log(nM))$-bit counter). This is due to the algorithm's
reliance on the identity $(\sum_{i=1}^n i\cdot z_i)^2 = (\sum_{i=1}^n
z_i)(\sum_{i=1}^n i^2\cdot z_i)$, for which there is no analogue
outside the strict turnstile setting. Using certain modifications
inspired by the Count-Min sketch \cite{cm05} it may be possible
to implement his algorithm in the turnstile model, though the resulting
space and time would be sub-optimal. In a different work, Ganguly
and Majumder
\cite{GM06} design a deterministic $k$-set structure based on
Vandermonde matrices, but the space required of this structure is
sub-optimal.

Our fast sketching algorithm for $L_1$-difference improves the running
time of an algorithm of Jayram and the second author \cite{jw08} 
by a $\Theta(\eps^{-2})$ factor for
estimating $L_1(L_2)$ of a matrix $A$, defined as the
sum of Euclidean lengths of the rows of $A$. As $L_1$-difference is a basic
primitive, we believe our algorithm is likely to have many further
applications.

\section{Preliminaries}\SectionName{prelim}
All space bounds mentioned throughout this paper are in bits, and all
logarithms are base $2$,  unless
explicitly stated otherwise. Running times are measured as the number
of standard machine word operations (integer arithmetic, bitwise
operations, and bitshifts). Each machine word is assumed to be
$\Omega(\log(nM))$ bits so that we can index each vector and do
arithmetic on vector entries in constant time. Also, for integer $A$,
$[A]$ denotes the set $\{1,\ldots, A\}$.

We now formally define the model in which our sketching procedure
runs. Alice receives $x\in\{-M,\ldots,M\}^n$, and Bob receives
$y\in\{-M,\ldots,M\}^n$. Both parties have access to a shared source of
randomness and must, respectively, output bit-strings $s(x)$ and
$s(y)$.
The requirement is that a third party can, given access to only $s(x)$
and $s(y)$, compute a value $Z$ such that $\Pr[|Z-||x-y||_1| >
\eps||x-y||_1] \le 1/3$ (recall $||x-y||_1\eqdef \sum_{i=1}^n |x_i -
y_i|$). The probability is over the randomness shared by Alice and
Bob, and the value $\eps\in (0,1]$ is a parameter given to all
parties. The goal is to minimize the lengths of $s(x)$ and $s(y)$, as
well
as the amount of time Alice and Bob each take to compute them.
Without loss of generality, throughout this
document we assume $x_i,y_i\ge 0$ for all $i$.  This promise can be
enforced by increasing all coordinates of $x,y$ by $M$, which
does not alter $||x-y||_1$.  Doing so increases the upper bound on
coordinate entries by a factor of two, but this alters our algorithm's
running time and resulting sketch size
by subconstant factors.

Since we present our sketching algorithm as a streaming algorithm
in \Section{body}, we now introduce some streaming notation.
We consider a vector $f = (f_1,f_2,\ldots,f_n)$ that is updated in a
stream as follows. The stream has
exactly $2n$ updates $(i_1,v_1),\ldots,(i_{2n},v_{2n})\in
[n]\times\{-M,\ldots,M\}$. Each update $(i,v)$ corresponds to the
action $f_i\leftarrow f_i + v$.
For each $j\in[n]$, there are exactly two
stream updates $(i,v)$ with $i=j$.  If these two stream updates are
$(i_{z_1},v_{z_1}), (i_{z_2},v_{z_2})$, then at most
one of $v_{z_1},v_{z_2}$ is negative, and at most one is positive. The
nonnegative update corresponds to adding $x_i$ to $f_i$, and the
nonpositive update corresponds to subtracting $y_i$ from $f_i$ (recall
we assumed $x_i,y_i\ge 0$). We make no restriction on the
possible values for $z_1$ and $z_2$. That is, our algorithm
functions correctly even if the stream
presents us with an adversarial permutation of the $2n$ coordinates
$x_1,\ldots,x_n,y_1,\ldots,y_n$. At the end of the stream $||f||_1 =
||x-y||_1$, so our streaming algorithm must
approximate $||f||_1$. For Alice and Bob to use our streaming
algorithm for sketching, Alice runs the algorithm with updates
$(i,x_i)$ for each $1\le i\le n$, and Bob separately runs the
algorithm (using the same random bits) with updates $(i,y_i)$. The
sketches they produce are simply the contents of the algorithm's
memory at the end of the stream.  It is a consequence of how our
algorithm works that these sketches can be
efficiently combined by a third party to approximate $||f||_1$.

\section{Main Streaming Algorithm}\SectionName{body}
Throughout this section we assume $\eps\ge 1/\sqrt{n}$.  Otherwise, we
can compute $||f||_1$ exactly by keeping the entire vector in memory
using $O(n\log M) = O(\eps^{-2}\log M)$ space with constant update
time.

\subsection{Handling Small $L_1$}\SectionName{kset}
We give a subroutine \textsc{TwoLeveLEstimator}, described in
\Figure{two-level}, to
compute $L_1$ exactly when promised that $L_1\le k$ for some parameter
$k\ge 2$. We assume that integers polynomially large in $k$
fit in a machine word, which will be true in our use of this
subroutine later.

In \Figure{two-level}, we assume we have
already calculated a prime $p$ satisfying
\begin{equation}\EquationName{cdef}
C\le p\le 2C,\ C = 4\cdot (5\ceil{\log k} + 24)^2\cdot\ceil{k/\log k} + 1
\end{equation}
(the choice of $p$ will be justified later), along with a
generator $g$ for the multiplicative group $\mathbb{F}_p^*$.  We also
precalculate logarithm tables $T_1,T_2$ such that $T_1[i] = g^i\mod p$
and $T_2[x]
= \dlog(x)$, where $0\le i\le p-2$ and $1\le x\le p-1$. 
Here
$\dlog(x)$ is the discrete logarithm of $x$ (i.e. the $i\in \GF(p)$
such that $g^i \equiv x \mod p$).

The subroutine \textsc{TwoLevelEstimator} makes calls to the following
algorithm given in \cite{DORS08}.

\begin{theorem}[Dodis {\it et al.} {\cite[Lemma
    E.1]{DORS08}}]\TheoremName{dodis}
Let $p$ be prime and $r = (r_x)_{x\in\mathbb{F}_p^*}$ have at most $s$
non-zero
entries ($2s + 1 < p$). Given $\sum_{x\in\mathbb{F}_p^*} r_x x^i$
for $i\in[2s]$, there is an algorithm to recover
$\{(x,r_x)|r_x\neq 0\}$ which uses
$O(s^2 + s(\log s)(\log\log s)(\log p))$ field operations over
$\GF(p)$.\afterproof
\end{theorem}

\begin{figure*}
\begin{center}
\fbox{
\parbox{6.375in} {
\underline{Subroutine \textsc{TwoLevelEstimator}}:\\

\texttt{\hspace{.1in} // Compute $||f||_1$ exactly when promised
  $||f||_1\le k$. The value $p$ is as in \Equation{cdef}.}
\begin{enumerate}
\addtolength{\itemsep}{-1mm}
\item Set $t = 2\ceil{\log k} + 12$ and $s = 2t + \ceil{\log k}$.
  Pick a random $h_1:[n]\rightarrow
  [\ceil{k/\log k}]$ from a $t$-wise independent hash family 
  and a random $h_2:[n]\rightarrow [p - 1]$ from a
  pairwise independent family.
\item For each $j\in[\ceil{k/\log k}]$
  maintain $2s$ counters $X_1^j,\ldots,X_{2s}^j$ modulo
  $p$, initialized to $0$.
 \item Upon seeing stream update $(i,v)$, increment $X_z^{h_1(i)}$ by
  $v\cdot (h_2(i))^z$ for $z\in[2s]$.
\item At the stream's end, for each $j\in[\ceil{k/\log k}]$, attempt
  to recover the non-zero entries of an
  $s$-sparse vector $f_j=((f_j)_x)_{x\in\mathbb{F}_p^*}$
  satisfying $\sum_{x\in\mathbb{F}_p^*} ((f_j)_x) x^z = X_z^j$ for
  each $z\in [2s]$ using \Theorem{dodis}.
\item Define $\sigma:\GF(p)\rightarrow \mathbb{Z}$ to be
  such that $\sigma(\alpha)$ equals $\alpha$ if $\alpha \le p/2$, and
  equals $\alpha - p$ otherwise.\\ Output 
  $\sum_{j=1}^{\ceil{k/\log k}}\sum_{(f_j)_x\neq 0} |\sigma((f_j)_x)|$.
\end{enumerate}
}}
\end{center}
\caption{\textsc{TwoLevelEstimator} subroutine pseudocode}\FigureName{two-level}
\end{figure*}

The proof of correctness of \textsc{TwoLevelEstimator} relies in part
on the following lemma.

\begin{lemma}[Bellare and Rompel {\cite[Lemma
    2.3]{BR94}}]\LemmaName{br94}
Let $X_i\in [0,1]$, $1\le i\le n$, be $t$-wise independent for $t\ge4$
an even integer,
$X = \sum_{i=1}^n X_i$, and $A>0$.  Then $\Pr[|X - \E[X]| \ge A] \le
8\left(\frac{t\E[X] + t^2}{A^2}\right)^{t/2}$. \afterproof
\end{lemma}

\begin{theorem}\TheoremName{two-level}
Ignoring the space to store the hash functions $h_1,h_2$ and tables
$T_1,T_2$, the algorithm \textsc{TwoLevelEstimator}
uses $O(k\log k)$ bits of space.  The hash functions $h_1,h_2$ and
tables $T_1,T_2$ require an additional $O((\log
k)(\log n) + k\log^2 k)$ bits.  The time to
process a stream update is $O(\log k)$.  If $L_1 \le k$, the final
output value of \textsc{TwoLevelEstimator} 
equals $L_1$ exactly with probability at least $3/4$.
\end{theorem}
\begin{proof}
Aside from storing $h_1,h_2,T_1,T_2$,  the number of counters is
$2s\ceil{k/\log k} = O(k)$, each of size $O(\log p) = O(\log k)$ bits,
totaling $O(k\log k)$ bits.
The space to store $h_1$ is $O((\log k)(\log n))$,
and the space to store $h_2$ is $O(\log n)$ \cite{CW79}.  The
tables $T_1,T_2$ each have $p-1 = O(k\log k)$ entries, each
requiring $O(\log p) = O(\log k)$ bits. Processing a stream update
requires evaluating $h_1,h_2$, taking $O(\log k)$ time and $O(1)$
time, respectively \cite{CW79}. 

As for update time, each stream token requires updating $2s = O(\log k)$
counters (Step 3). Each counter update can be done in constant time
with the help of table lookup since $(h_2(i))^z = g^{z\cdot
  \dlog(h_2(i))} = T_1[(z\cdot T_2[h_2(i)])\mod (p-1)]$.

We now analyze correctness.  Define $I = \{i\in [n] : f_i\neq 0\hbox{
  at the stream's end}\}$.
Note $|I| \le L_1 \le k$.  For
$j\in[\ceil{k/\log k}]$, define the random variable $Z_j =
|h_1^{-1}(j)\cap I|$.  We now define two events.

Let $\mathcal{Q}$ be the event that $Z_j \le s = 2t + \ceil{\log k}$
for all $j\in[\ceil{k/\log k}]$.

Let $\mathcal{Q}'$ be the event that there do not exist distinct
$i,i'\in I$ with both $h_1(i) = h_1(i')$ and $h_2(i) = h_2(i')$.

We first argue that, conditioned on both $\mathcal{Q},\mathcal{Q}'$
holding, the output of \textsc{TwoLevelEstimator} is
correct. Note $p-1 \ge 4s^2\ceil{k/\log k} \ge 100k\log k$ (recall the
definition of $s$ in Step 1 of \Figure{two-level}). If $\mathcal{Q}'$
occurs,
$|h_2^{-1}(i)\cap h_1^{-1}(j)\cap I| \le 1$ for all $i\in [p-1]$ and
$j\in[\ceil{k/\log k}]$.  One can then view $X_z^j$ as holding
$\sum_{x\in\mathbb{F}_p^*}(r_j)_x x^z$, where $(r_j)_x$ is the frequency
(modulo $p$) of the unique
element in the set $h_2^{-1}(i)\cap h_1^{-1}(j)\cap I$ (or $0$ if that
set is empty).  Conditioned on $\mathcal{Q}$, every $r_j$ is
$s$-sparse, so we correctly recover $r_j$ in Step 4 by
\Theorem{dodis} since $2s + 1 =
5\ceil{\log k} + 13 < 100k\ceil{\log k} < p$.  Note that $p$ is
strictly greater than twice the absolute value of the largest
frequency since $L_1 \le k$, and thus negative frequencies
are strictly above $p/2$ in $\GF(p)$, and positive frequencies are
strictly below $p/2$. Thus, given that the $r_j$
are correctly recovered, $\sigma$ correctly recovers the
actual frequencies in Step 5, implying correctness of the final
output.

Now we proceed to lower bound $\Pr[\mathcal{Q}\wedge
\mathcal{Q}']$. First
we show $\mathcal{Q}$ occurs with probability at least $7/8$.
If we let $Z_{j,i}$ indicate $h_1(i) = j$, then note the random
variables $\{Z_{j,i}\}_{i\in I}$ are $t$-wise independent and $Z_j =
\sum_{i\in I} Z_{j,i}$. Also, $\E[Z_j] = |I|/\ceil{k/\log k} \le \log k$.  
Noting 
$$\Pr[\neg\mathcal{Q}] \le \Pr[|Z_j - \E[Z_j]| \ge 2t]$$
then setting $A = 2t$ and applying \Lemma{br94},
\begin{eqnarray*}
\Pr[|Z_j - \E[Z_j]| \ge 2t] &\le&
8\left(\frac{t\E[Z_j] + t^2}{(2t)^2}\right)^{t/2}\\
&\le& 8\left(\frac{2t^2}{4t^2}\right)^{\log k + 6} \le \frac{1}{8k}
\end{eqnarray*}
since $\E[Z_j] \le t$.  A union bound implies $\Pr[\mathcal{Q}] \ge
7/8$.

Now we analyze $\Pr[\mathcal{Q}'|\mathcal{Q}]$.  Let $Y_{i,i'}$ be a
random variable indicating $h_2(i) = h_2(i')$ and define the random
variable $Y =
\sum_{(i,i')\in\binom{I}{2},h_1(i)=h_1(i')} Y_{i,i'}$. Note
$\mathcal{Q'}$ is simply the event that $Y = 0$.  We have
\begin{eqnarray*}
\E[Y] &=& \sum_{j=1}^{\ceil{k/\log k}}\E\left[
  \sum_{(i,i')\in\binom{h_1^{-1}(j)\cap I}{2}}
  \Pr[h_2(i)=h_2(i')]\right]\\
&\le& \sum_{j=1}^{\ceil{k/\log k}} \frac{\E[|h_1^{-1}(j)\cap
  I|^2]/2}{p-1}\\
&\le&
\sum_{j=1}^{\ceil{k/\log k}} \frac{\E[|h_1^{-1}(j)\cap I|^2]/2}{4s^2\ceil{k/\log
      k}}
\end{eqnarray*}
where the expectation on the right side of the first equality is
over the random choice of
$h_1$, and the probability is over the random choice of $h_2$.  The
first inequality holds by pairwise independence of $h_2$.
Conditioned on $\mathcal{Q}$, $|h_1^{-1}(j)\cap I| \le s$ for all
$j$ so that $\E[Y|\mathcal{Q}] \le 1/8$, implying
$\Pr[\mathcal{Q'}|\mathcal{Q}] = 1 - \Pr[Y
\ge 1|\mathcal{Q}] \ge 7/8$ by Markov's Inequality.

In total, we have $\Pr[\mathcal{Q}\wedge\mathcal{Q}'] =
\Pr[\mathcal{Q}]\cdot \Pr[\mathcal{Q}'|\mathcal{Q}] \ge (7/8)^2 >
3/4$, and the claim is proven.
\end{proof}

\begin{remark}\RemarkName{hash-stuff}
In Step 1 of \Figure{two-level}, we twice pick a hash function
$h:[a]\rightarrow [b]$ from an $m$-wise independent family for
some integers $m$ and $a\neq b$ (namely, when picking $h_1$ and
$h_2$). However, known
constructions \cite{CW79} have $a=b$, with $a$ a prime power.  This is
easily circumvented.
When we desire an $h$ with unequal domain size $a$ and range size $b$,
we can pick a prime $\ell\ge2\cdot
\max\{a,b\}$ then pick an $m$-wise independent hash function
$h':[\ell]\rightarrow [\ell]$ and define $h(x) \eqdef (h'(x)\mod b) +
1$. The family of such $h$ is still $m$-wise independent, and by
choice of $\ell$, no range value is more than twice more likely than any
other, which suffices for our application with a slight worsening of
constant factors.
\end{remark}

The following theorem analyzes the pre-processing and post-processing
complexity of \textsc{TwoLevelEstimator}.

\begin{theorem}\TheoremName{tle-preprocess}
Ignoring the time needed to find the prime $\ell$ in \Remark{hash-stuff},
the pre-processing time of \textsc{TwoLevelEstimator} before
seeing the stream is $O(k\log k)$, and the post-processing time is
$O(k\log k\log\log k\log\log\log k)$.
\end{theorem}
\begin{proof}
We first discuss the pre-processing time.  It is known that the prime
$p$ and generator $g$ for $\mathbb{F}_p^*$ can be found in time
$\polylog(C) = \polylog(k)$ (see the proof of Theorem 4 in
\cite{BM84}). Once we have $p,g$, filling in
$T_1,T_2$ takes $O(p) = O(k\log k)$ time, which dominates the
pre-processing time.  The time to allocate the
$O(k)$ counters $X_z^j$ is just $O(k)$.

The post-processing work is done in Steps 4
and 5 in \Figure{two-level}. For Step 4, there are $O(k/\log k)$
values of $j$, for each of which we run the algorithm of
\Theorem{dodis} with $s = O(\log k)$ and $p = O(k\log k)$, thus
requiring a total of $O(k\log k\log\log k\log\log\log k)$ field
operations over $\GF(p)$.  Since we precalculate the table $T_2$, we
can do all $\GF(p)$ operations in constant time, including
division. In
Step 5 we need to sum the absolute values of $O(\log k)$ non-zero
entries of $O(k/\log k)$ vectors $f_j$, taking time $O(k)$.
\end{proof}

\begin{remark}
Our subroutine \textsc{TwoLevelEstimator} uses the fact that since
$L_1\le k$ and $f$ is an integer vector, it must be the case that $L_0
\le k$.  From here, what we develop is a {\em $k$-set structure} as
defined by Ganguly \cite{Ganguly07}, which is a data structure
that allows one to recover the $k$-sparse vector $f$.  In fact, any
$k$-set structure operating in the turnstile model (i.e., where some
$f_i$ can be negative) would have sufficed in place of
\textsc{TwoLevelEstimator}.  We develop our
particular subroutine since previous approaches were either less
space-efficient or did not work in the turnstile setting
\cite{GM06,Ganguly07}.  We remark that at the cost of an
extra $O(\log^2 k)$ factor in space, but with the benefit of only
$O(1)$ post-processing time, one can replace
\textsc{TwoLevelEstimator} with an alternative scheme.  Namely, for each
$j\in [\ceil{k/\log k}]$, attempt to perfectly the hash the $O(\log
k)$ coordinates contributing to $||f||_1$ mapped to $j$ under $h_1$ by
pairwise independently hashing into
$O(\log^2 k)$ counters, succeeding with
constant probability.  Each counter holds frequency sums modulo $p$.
By repeating $r=\Theta(\log k)$
times and taking the maximum sum of counter absolute values over any
of the $r$ trials, we succeed in finding the sum of frequency absolute
values of items mapping to $j$ under $h_1$ with probability $1 -
1/\poly(k)$. Thus
by a union bound, we recover $||f||_1$ with probability $99/100$ by
summing up over all $j$.  The
estimate of $||f||_1$ can be maintained on the fly during updates to
give $O(1)$ post-processing, and
updates still take only $O(\log k)$ time.
\end{remark}

\subsection{The Full Algorithm}\SectionName{full-stream}
Our full algorithm requires, in part, a constant factor approximation
to the $L_1$-difference.  To obtain this, we can use the algorithm of
Feigenbaum {\it et al.} \cite{FKSV02} with $\eps$ a constant.

\begin{theorem}[Feigenbaum {\it et al.} {\cite[Theorem
    12]{FKSV02}}]\TheoremName{feig}
There is a one-pass streaming algorithm for $(1\pm\eps)$-approximating
the $L_1$-difference using $O(\eps^{-2}\log(nM))$
space with update time $O(\eps^{-2}\log(nM))$, and succeeding with
probability at least $19/20$.\afterproof
\end{theorem}

\begin{remark}
It is stated in \cite{FKSV02} that the update time
in \Theorem{feig} is $O(\eps^{-2}\mathrm{field}(\log(nM)))$,
where $\mathrm{field}(D)$ is the time to do arithmetic over
$\GF(2^D)$ (not including division).  Section 2.2 of
\cite{FKSV02} points out that
$\mathrm{field}(D) = O(D^2)$ na\"{i}vely. In fact, it suffices for the
purposes of their algorithm to work over $\GF(2^D)$ for the smallest
$D\ge\log(nM)$ such that $D = 2\cdot 3^{\ell}$, in which case a highly
explicit irreducible polynomial of degree $D$ over $\mathbb{F}_2[x]$
(namely $x^D + x^{D/2} + 1$ \cite[Theorem 1.1.28]{vanLint99})
can be used to perform $\GF(2^D)$ arithmetic in time
$O(D)$ in the word RAM model without any additional pre-processing
space or time.
\end{remark}

We also make use of the following algorithm due to Pavan and
Tirthapura \cite{PavanTir07}.

\begin{theorem}[Pavan and Tirthapura {\cite[Theorem
    2]{PavanTir07}}]\TheoremName{pavan}
Let $a,b,c,d,x,r,m$ be integers fitting in a machine word with $m>0$
and $a,b,c,d\in\{0,\ldots,m-1\}$.  There is an algorithm to
calculate $|\{i \ :\ (a\cdot(x+i) + b\mod m)\in [c, d],\ 0\le i\le r\}|$ in
time $O(\log(\min(a, r)))$ using $O(\log(r\cdot m))$
space.\afterproof
\end{theorem}

\begin{figure*}
\begin{center}
\fbox{
\parbox{6.375in} {
\underline{Main Algorithm \textsc{L1-Diff}}:
\begin{enumerate}
\addtolength{\itemsep}{-1mm}
\item Set $\eps' = \eps/8$.
\item Pick a random hash function $h:[q]\rightarrow [q]$ from a
  pairwise independent family so that $h(x) = ax+b\mod q$ for some
  prime $q\in[2nM, 4nM]$ and $a,b\in\GF(q)$.
\item Initialize instantiations $\TLE_1,\ldots,\TLE_{\ceil{\log((\eps')^2
      nM)}}$ of \textsc{TwoLevelEstimator} with $k=\ceil{4/(\eps')^2}$.
  All instantiations share the same prime $p$, generator $g$, hash
  functions
  $h_1,h_2$, and logarithm tables $T_1,T_2$.
\item Upon seeing stream update $(i,v)$, let $v_j$ be the output of
  the algorithm from \Theorem{pavan} with inputs $a,b$ as in Step 2,
  $c=c_j=2^{\floor{\log q} - j}$, $d = d_j=2^{\floor{\log q} - j + 1} - 1$,
  $x = (i-1)M + 1$, $r = |v| - 1$, and $m = q$.  Feed the update
  $(i,\mathrm{sgn}(v)\cdot v_j)$ to $\TLE_j$ for
  $j=1,\ldots,\ceil{\log((\eps')^2 nM)}$. Let $R_j$ be the output of
  $\TLE_j$.
\item Run an instantiation $\TLE$ of \textsc{TwoLevelEstimator} in
  parallel with $k=\ceil{1/(\eps')^2}$
  which receives all updates, using the same $h_1,h_2,p,g,T_1,T_2$ of
  Step 2. Let its output be $R$.
\item Run the algorithm of \Theorem{feig} in parallel with error
  parameter $1/3$ to obtain a value $R'\in[L_1/2,L_1]$.
\item If $R' \le \ceil{1/(\eps')^2}$, output $R$. Otherwise, output
  $q\cdot 2^{\ceil{\log((\eps')^2 R')}-\floor{\log
      q}}R_{\ceil{\log((\eps')^2 R')}}$.
\end{enumerate}
}}
\end{center}
\caption{\textsc{L1-Diff} pseudocode}\FigureName{fullalg}
\end{figure*}

Our main algorithm, which we call \textsc{L1-Diff}, is described in
\Figure{fullalg}.  Both in \Figure{fullalg} and in the proof of
\Theorem{main-alg}, $\mathrm{sgn}$ denotes the function which takes as
input a real number $x$ and outputs $-1$ if $x$ is negative, and $1$
otherwise.

\begin{theorem}\TheoremName{main-alg}
The algorithm \textsc{L1-Diff} has update time $O(\log(\eps^2
nM)\log(M/\eps))$ and the space used is
$O(\eps^{-2}\log(1/\eps)\log(\eps^{-2}nM))$.
Pre-processing requires
$\polylog(nM) + O(\eps^{-2}\log(1/\eps)\log(\eps^{-2}nM))$ time.
Time $O(\eps^{-2}\log(1/\eps)\log\log(1/\eps)\log\log\log(1/\eps))$ is
needed for post-processing.
The output is $(1\pm\eps)L_1$ with probability at least $2/3$.
\end{theorem}
\begin{proof}
The hash function $h$ requires $O(\log(nM))$ space.  There are
$O(\log(\eps^2 nM))$ instantiations of \textsc{TwoLevelEstimator}
(Steps 2 and 4), each with $k = O(\eps^{-2})$, taking a total of
$O(\eps^{-2}\log(1/\eps)\log(\eps^2 nM))$ space by
\Theorem{two-level}. The hash functions $h_1,h_2$ and tables $T_1,T_2$
take $O(\log(1/\eps)\log(n) + \eps^{-2}\log^2(1/\eps)) =
O(\eps^{-2}\log(1/\eps)\log n)$ space, also by \Theorem{two-level}
(recall we assume $\eps \ge 1/\sqrt{n}$). Step 6 requires only
$O(\log(nM))$ space by \Theorem{feig}, since the algorithm is run
with error parameter $1/3$.

As for running time, in Step 3 we call the algorithm of
\Theorem{pavan} $O(\log(\eps^2 nM))$ times, each time with $a < q$ and
$r \le M$, thus taking a total of $O(\log(\eps^2 nM)\log(\min(q, M)))
= O(\log(\eps^2 nM)\log M)$ time. We must also feed the necessary
update to each $\TLE_j$, each time taking $O(\log(1/\eps))$ time by
\Theorem{two-level}.  Updating every $\TLE_j$ thus takes time
$O(\log(\eps^2 nM)\log(1/\eps))$.

In pre-processing we need to pick a prime $q$ in the desired range,
which can be accomplished by picking numbers at random and testing
primality; the expected time is $\polylog(nM)$. We also need to
prepare $h_1,h_2,T_1,T_2$ and all the \textsc{TwoLevelEstimator}
instantiations, which takes $O(\eps^{-2}\log(1/\eps)\log(\eps^2 nM))$
time by \Theorem{tle-preprocess}, in addition to the $\polylog(n)$
time required to find an appropriate prime $\ell$ as described in
\Remark{hash-stuff}. The pre-processing time for Step 6
is $O(1)$ (see Figure 1 of \cite{FKSV02}).

In post-processing we need to recover the estimate $R'$ from Step 6,
which takes $O(1)$ time, then recover an estimate from some
\textsc{TwoLevelEstimator} instantiation, so the time is as
claimed. In post-processing, to save time one should not run
Steps 4 and 5 of \textsc{TwoLevelEstimator} in \Figure{two-level}
except at the instantiation whose output is used in Step 7.

Now we analyze correctness. Let $\mathcal{Q}$ be the event that
$R'\in [L_1/2,L_1]$. We proceed by a case analysis.

For the first
case, suppose $L_1 \le
\ceil{1/(\eps')^2}$.  Then, $\TLE$
computes $L_1$ exactly with probability at least $3/4$ by
\Theorem{two-level}, and hence overall we output $L_1$ exactly with
probability at least $(19/20)\cdot (3/4) > 2/3$.

Now, suppose $L_1 > \ceil{1/(\eps')^2}$. In analyzing this case, it helps
to view \textsc{L1-Diff} as actually computing
$L_0(f') \eqdef |\{i : f_i' \neq 0\}|$, where we consider an
$nM$-dimensional vector $f'$ that is being updated as follows: when
receiving an update $(i,v)$ in the stream, we conceptually view this
update as being $|v|$ updates
$((i-1)M+1,\mathrm{sgn}(v)),\ldots, ((i-1)M+|v|,\mathrm{sgn}(v))$ to
the vector $f'$. Here, the vector $f'$ is initialized to
$\vec{0}$. Note that at the stream's end, $L_0(f') = ||f||_1$.

Let $f'^j$ denote the vector whose $i$th entry, $i\in[mM]$, is $f'_i$
if $h(i)\in [c_j,d_j]$ and $0$ otherwise.  That is, $f'^j$ receives
stream updates
only from items fed to $\TLE_j$. For $i\in[nM]$, let $X_{i,j}$ be a
random variable indicating $h(i)\in [c_j, d_j]$, and let $X_j =
\sum_{f'_i\neq 0} X_{i,j}$ so that $X_j = L_0(f'^j)$. Define $p_j
\eqdef (d_j-c_j+1)/q =
2^{\floor{\log q} - j}/q$ so that $\E[X_{i,j}] = p_j$. Thus, $\E[X_j]
= p_j\cdot L_0(f')$. 
Note that $1/2\le 2^{\floor{\log q}}/q\le 1$.  Conditioned on
$\mathcal{Q}$, we have the inequalities
$$\frac{L_0(f')}{2^{\ceil{\log((\eps')^2
    R')}}} \le \frac{L_0(f')}{(\eps')^2 R'} \le
\frac{2}{(\eps')^2}$$
and
$$\frac{L_0(f')}{2^{\ceil{\log((\eps')^2 R')}}} \ge
\frac{L_0(f')}{2(\eps')^2 R'} \ge \frac{1}{2(\eps')^2}$$
By
the choice of $j=\ceil{\log((\eps')^2 R')}$ in
Step 7 of \Figure{fullalg}, we thus have, assuming $\mathcal{Q}$
occurs,
$$\frac{16}{\eps^2} = \frac{1}{4(\eps')^2} \le  \E[X_j]
\le \frac{2}{(\eps')^2}$$
since $\E[X_j]
= p_j\cdot L_0(f') = (2^{\floor{\log q}}/q)\cdot (L_0(f')/2^j)$.

Let $\mathcal{Q}'$ be the event that $|X_j - \E[X_j]| \le \eps
\E[X_j]$.  Applying Chebyshev's inequality,
$$\Pr[\mathcal{Q}'|\mathcal{Q}] \ge 1 -
\frac{\Var[X_j]}{\eps^2\E^2[X_j]} \ge
1 - \frac{1}{\eps^2\E[X_j]} \ge \frac{15}{16} $$
The second inequality holds since $h$ is pairwise independent and $X_j$
is the sum of Bernoulli random variables, implying $\Var[X_j] = \sum_i
\Var[X_{i,j}] \le \sum_i \E[X_{i,j}] = \E[X_j]$. The last inequality
holds by choice of $\eps' = \eps/8$.

Let $\mathcal{Q}''$ be the event that $\TLE_j$ outputs $X_j$
correctly. Now, conditioned on
$\mathcal{Q}\wedge\mathcal{Q}'$, we have $X_j \le 2(1+\eps)/(\eps')^2
\le
4/(\eps')^2$ since $\eps \le 1$.  Thus by \Theorem{two-level},
$\Pr[\mathcal{Q}''|\mathcal{Q}\wedge \mathcal{Q}'] \ge 3/4$. Overall,
we compute
$L_1$ of the entire stream correctly with probability at least
\begin{eqnarray*}
\Pr[\mathcal{Q}\wedge \mathcal{Q}' \wedge \mathcal{Q}''] &=&
\Pr[\mathcal{Q}]\cdot \Pr[\mathcal{Q}'|\mathcal{Q}] \cdot
\Pr[\mathcal{Q}''|\mathcal{Q}\wedge \mathcal{Q}']\\
&\ge& (19/20)\cdot
(15/16)\cdot (3/4) > 2/3
\end{eqnarray*}
\end{proof}

Our streaming algorithm also gives a sketching procedure. This is
because, as long as Alice and Bob share randomness, they can generate
the same $h,h_1,h_2,p,g$ then separately apply the streaming
algorithm to their vectors $x,y$. The sketch is then just the state of
the
streaming algorithm's data structures.  Since each stream token causes
only linear updates to counters, a third party
can then take the counters from Bob's sketch and subtract them from
Alice's, then do post-processing to recover the estimation of the
$L_1$-difference. The running time for Alice and Bob to produce their
sketches is the streaming algorithm's pre-processing time, plus $n$
times the update time. The time for the third party to obtain an
approximation to $||x-y||_1$ is the time required to combine the
sketches, plus the post-processing time.  We thus
have the following theorem.

\begin{theorem}\TheoremName{sketching}
Sharing $\polylog(nM)$ randomness, two parties Alice and Bob,
holding vectors $x,y\in\{-M,\ldots,M\}^n$, respectively, can produce
$O(\eps^{-2}\log(1/\eps)\log(\eps^2 nM))$-bit sketches $s(x),s(y)$
such that a third party can recover $||x-y||_1$ to within $(1 \pm \eps)$
with probability at
least $2/3$ given only $s(x),s(y)$.  Each of Alice and Bob use
time $O(n\log(\eps^2 nM)\log(M/\eps))$ to produce their sketches. In
$O(\eps^{-2}(\log(\eps^2 nM)
+\log(1/\eps)\log\log(1/\eps)\log\log\log(1/\eps)))$ time,
the third party can recover $||x-y||_1$ to within a multiplicative
factor of $(1 \pm \eps)$. \afterproof
\end{theorem}

Note Alice and Bob's running time is
always $O(n\log^2(nM))$ since $\eps\ge 1/\sqrt{n}$.

\begin{remark}
Though we assume Alice and Bob share randomness, to actually implement
our algorithm in practice this randomness must be communicated at some
point.
We note that while the sketch length guaranteed by \Theorem{sketching}
is $O(\eps^{-2}\log(1/\eps)\log(\eps^2 nM))$ bits, the required amount
of shared randomness is $\polylog(nM)$, which for large enough $\eps$
is larger than the sketch length.  This is easily fixed
though. Since the required randomness is only polynomially larger than
the space used by the sketching algorithm (which is asymptotically equal
to the sketch length), the two parties can use the Nisan-Zuckerman
pseudorandom
generator \cite{NisanZu96} to stretch a seed whose length is linear
  in the sketch length to a pseudorandom string of length
$\polylog(nM)$ which still provides the guarantees of
\Theorem{sketching}.  Alice and Bob then only need to communicate this
random seed.
\end{remark}

\section*{Acknowledgments}
We thank Avinatan Hassidim, Piotr Indyk, Yuval Ishai, and Swastik
Kopparty for useful comments and discussions, and Milan Ru\v{z}i\'{c}
for pointing out a tweak to
an early version of our work which improved our space by a factor of
$O(\log(1/\eps))$. We thank Venkat Chandar for pointing out the
reference \cite{Wyner74}, and Silvio Micali for pointing out
\cite{BM84}.

\bibliographystyle{plain}

\bibliography{./allpapers}

\end{document}